\documentclass[manuscript,11pt]{acmart} 
\renewcommand\footnotetextcopyrightpermission[1]{} 
\usepackage{wrapfig}

\newcommand{\remove}[1]{}

\usepackage[vlined,ruled,linesnumbered,noresetcount]{algorithm2e}
\DontPrintSemicolon

\begin{document}
\title{Brief Announcement: Shallow Overlay Trees Suffice for High-Throughput Consensus}

\author{Wojciech Golab}
\authornote{Wojciech Golab is supported in part by the Natural Sciences and Engineering Research Council (NSERC) of Canada, Discovery Grants Program, and by a Google Faculty Research Award.}
\affiliation{%
  \institution{Department of Electrical and Computer Engineering University of Waterloo}
  \streetaddress{200 University Avenue West, Waterloo, Ontario, N2L 3G1, Canada}
}
\email{wgolab@uwaterloo.ca}

\author{Hao Tan}
\affiliation{%
  \institution{School of Computer Science University of Waterloo}
  \streetaddress{200 University Avenue West, Waterloo, Ontario, N2L 3G1, Canada}
}
\email{h26tan@uwaterloo.ca}

\begin{abstract}
 All-to-all data transmission is a typical data transmission pattern in blockchain systems. Developing an optimization scheme that provides high throughput and low latency data transmission can significantly benefit the performance of those systems. In this work, we consider the problem of optimizing all-to-all data transmission in a wide area network (WAN) using overlay multicast. We prove that in a congestion-free core network model, using shallow broadcast trees with heights up to two is sufficient for all-to-all data transmission to achieve the optimal throughput allowed by the available network resources.
\end{abstract}
\maketitle

\section{Introduction}
 Starting with Bitcoin's great success \cite{Bitcoin}, blockchain systems have become a critical use case for high-performance consensus protocols. A practical blockchain system can involve hundreds to thousands of geographically distributed nodes. In many blockchain systems, all to all is the dominating communication pattern. Each participating node has to receive all other nodes' transactions to achieve complete decentralization in transaction processing. Also, the latency of committing a transaction is determined by when the sender receives acknowledgements from all other nodes or a majority quorum of nodes. To enhance the performance of blockchain systems, this paper considers the following general problem. Given a set of nodes connected by a WAN and each node having a stream of data to broadcast to all other nodes, how can we maximize the aggregated broadcast throughput while minimizing the latency for each node's data to reach all other nodes.

Previous work \cite{SteinerTreePacking} has shown that optimizing multicast throughput alone is already an NP-Hard problem with the constraints of network topology and link capacities. Efficient solutions usually rely on heuristics or advanced technique such as network coding \cite{NetworkCoding}. However, those solutions are not practical in a WAN environment due to the opaqueness of WAN's exact physical topology. Instead of assuming a specific network topology for WAN, we model the WAN as a set of sites connected by a core network with unlimited bandwidth. Every site can send and receive data from each other, bottle-necked only by the edge link bandwidth between each site and the core network. This model is not only simple but also valid as shown by recent measurements \cite{InternetCongestion}.

Moreover, overlay multicast has been proven to be a useful technique to provide high throughput data dissemination in a network without detailed topological information. To our surprise, when modelling the network as a congestion-free core and leveraging overlay multicast, it is possible for all-to-all data transmission to achieve both optimal aggregated throughput and low latency for data to reach all other nodes. The general idea is similar to \cite{SplitStream, Bullet}, which split the source stream at each node into multiple partitions and broadcast each partition with different overlay trees. We build on this prior work by proving that the optimal aggregated throughput is always achievable by using broadcast trees with heights up to two.

\section{Preliminaries}

\textbf{Network Model.} The network topology is a complete graph $G = (V, E)$ with $n$ vertices. There are two functions $ C_u: V \rightarrow \mathbb{R}^+$ and $ C_d: V \rightarrow \mathbb{R}^+$, which respectively define the uplink and downlink bandwidth of the edge link from a site to the core network. The data transmission between two sites consumes the sender's $uplink$ bandwidth and the receiver's $downlink$ bandwidth. The uplink and the downlink bandwidth of a site are shared by all unicast data transmissions associated with the site. For instance, a sender multicasting to $n$ receivers at the rate of $R$ will consume $nR$ of the sender's uplink bandwidth and $R$ of each receiver's downlink bandwidth.

\textbf{Rate of client data streams.} A client data stream is an infinite sequence of data batches from clients to be broadcast to all other sites in the network. The rate $R_i$ of a client data stream $s_i$  represents the incoming rate of client data at site $v_i$. For instance, letting $r$ be the number of incoming client requests per second at site $v$ and letting $S$ be the size of the data payload of each request, then the client data rate at site $v$ is $rS$. We assume that a site's client data stream does not consume its downlink bandwidth.

\textbf{Sustainable Rates.} For an all-to-all data transmission in $G(V,E)$ with $n$ sites, each site is associated with a client data stream $s_i$ that must be received by all other sites. Client data streams $s_1, s_2, \ldots, s_n$ with rates $R_1, R_2, .., R_n$ are said to be \emph{sustainable} if the following three conditions are all met: \\
(1) For each $v_i \in V$, $R_i \leq C_u(v_i)$; \
(2) For each $v_i \in V$, $\sum_{j \neq i} R_j \leq C_d(v_i)$; \ 
(3) $(n-1)\sum_{i=1}^{n} R_i \leq \sum_{i=1}^{n} C_u(v_i)$.

Intuitively, being sustainable is the minimum requirement for a set of client data streams to be broadcast at their incoming rates. Condition (1) ensures that each site has enough uplink bandwidth to send out its data at least once to other nodes. As each site has to receive from all other peers, condition (2) ensures that the aggregated rate of incoming streams does not exceed a site's downlink bandwidth. Condition (3) derives from the fact that the client data at each site must be sent at least $n-1$ times. If any of the above conditions are violated, the aggregated throughput of all-to-all data transmission will be less than $(n-1)\sum_{i=1}^{n} R_i$.

\textbf{Partitioning Scheme.} Assume the rate of a client data stream can be split at any granularity. A partitioning scheme $P(s_i, n)$ of a client data stream $s_i$ with rate $R_i$ splits elements of $s_i$ into $n$ streams $s_{i,1}, \ldots, s_{i,n}$ with rates $r_{i,1}, \ldots, r_{i,n}$ such that $R_i = \sum_{j = 1}^{n} r_{i,j}$. We refer to each split of the stream a \textit{sub-stream} of $s_i$. 

\section{Main Results}
\begin{theorem}
\label{th:1}
For client data streams $s_1, \ldots, s_n$ with sustainable rates  $R_1, \ldots, R_n$, there exists a partitioning scheme for each client data stream such that:
\begin{enumerate}
    \item Each sub-stream can be broadcast at its rate without violating downlink and uplink bandwidth constraints at any site.
    \item The height of each sub-stream's broadcast tree is at most 2.
\end{enumerate}
\end{theorem}

\subsection{Proof Sketch}
We prove theorem \ref{th:1} by constructing a possible partitioning scheme for each client data stream and associating each sub-stream with a broadcast tree with a height up to two.

\subsubsection{Constructing Sub-streams}
\label{section:construct-substreams}
Each client stream $s_i$ will be split into $n$ sub-streams $s_{i,1}, \ldots, s_{i,n}$ with rates $r_{i,1}, \ldots, r_{i,n}$. The data of the special sub-stream $s_{i,i}$ is sent directly from $v_i$ to all the remaining sites. The data of sub-stream $s_{i,j}$ for $i \neq j$ is sent from $v_i$ to $v_j$ first, and then $v_j$ will broadcast the data to the rest of the sites. All the broadcast trees defined previously have a height bounded by two.

\begin{wrapfigure}{r}{0.33\textwidth}
\begin{algorithm}[H]
 \label{algo:1}
 \SetKwInOut{Input}{Input}
 \SetKwInOut{Output}{Output}
 \Input{$G(V,E)$ \newline $C_u : V \rightarrow \mathbb{R}+$ \newline $R_i \ldots R_n$}
 \Output{$r_{1,1} \ldots r_{n,n}$}
  $r_{i,j} := 0,  1 \leq i,j \leq n$\;
  $U_i := C_u(v_i) - R_i, 1 \leq i \leq n$\;\label{line:2}
 \For{$i \gets$ 1 \textbf{to} $n$}{
  $R^{\prime}_i := R_i$\;
  \For{$j \gets$ 1 \textbf{to} $n$}{
    \eIf{$(n-2)R^{\prime}_i > U_j $}{\label{line:7}
        $r_{i,j} := \frac{U_j}{n-2}$\; \label{line:8}
    }{
        $r_{i,j} := R^{\prime}_i$ \label{line:10}
    }
    $U_j := U_j - (n-2)r_{i,j}$\;\label{line:11}
    $R^{\prime}_i :=  R^{\prime}_i - r_{i,j} $\; \label{line:13}
    \If{$R^{\prime}_i = 0$}{
        \textbf{break}\;
    }
  }
 }
  \Return $r_{1,1} \ldots r_{n,n}$\;
\caption{Sub-stream rate assigning algorithm}
\end{algorithm}
\vspace{-24pt}
\end{wrapfigure}

\subsubsection{Computing Sub-stream Rates}
Algorithm \ref{algo:1} computes the rate of each sub-stream defined in the previous section. For each node $v_i$, we divide its uplink bandwidth into two parts: $U_i^{\prime} = R_i$ and $U_i = C_u(v_i) - R_i$. $U_i^{\prime}$ represents the reserved uplink bandwidth for $v_i$ to send out all its data at least once; $U_i$ is the residual uplink bandwidth after excluding reserved uplink bandwidth.  The algorithm will iterate over all site pairs in $\{(i,j) | 1 \leq i \leq n, 1 \leq j \leq n\}$ in lexicographical order to compute sub-stream rates. For the iteration when sub-stream rate $r_{i,j}$ is computed,  the algorithm greedily allocates as much of $U_j$ as possible to $r_{i,j}$ until either $U_j$ is exhausted or the aggregated sub-stream rate reaches $R_i$. According to the overlay trees defined in the previous section, sending $s_{i,j}$ consumes $r_{i,j}$ of $U_i^{\prime}$ and $(n-2)r_{i,j}$ of $U_j$. This rule also applies to the case $i = j$, where sending $s_{i,i}$ consumes $(n-1)r_{i,i}$ of $U_i$. Note that, Algorithm~\ref{algo:1} does not aim to compute the optimal partitioning scheme, which favours one level tree overlays.



\section{Conclusion and Discussion}
According to theorem \ref{th:1}, by leveraging overlay multicast,   all-to-all data transmission can achieve the best possible throughput without paying an expensive price for latency. This result provides a theoretical foundation for ruling out deep overlay trees with heights greater than two when optimizing all-to-all data transmission for applications such as blockchains and consensus protocols. Although Theorem \ref{th:1} relies on the rate of incoming data to be sustainable, we can resort to a two phase optimization when dealing with client data rates that are not sustainable. The first phase computes the optimal sustainable rates based on available network resources. The second phase computes the optimal combinations of overlay trees that yield the lowest latency given the sustainable rates obtained in the first phase. 

\bibliographystyle{ACM-Reference-Format}
\bibliography{sample}

\newpage
\appendix
\section{Pictures}
\begin{figure}[htbp]
  \centering
  \includegraphics[width=\linewidth]{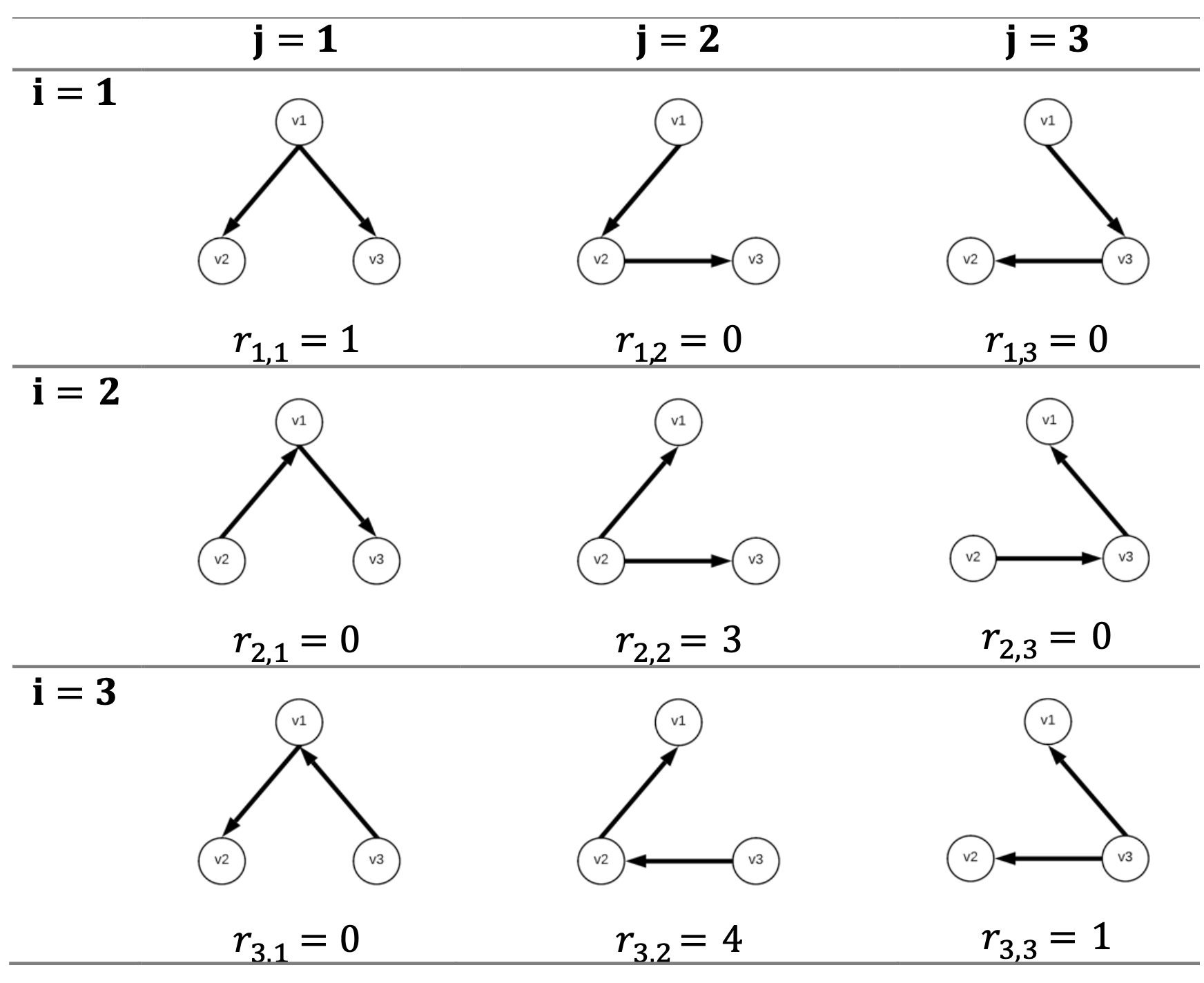}
  \caption{\textbf{Visualization of algorithm \ref{algo:1}}: This figure demonstrates an example with three sites $v_1, v_2, v_3$  with uplink bandwidth $C_u(v_1) = 2$, $C_u(v_2) = 10$, $C_u(v_3) = 6$. Let client data stream rates be $R_1 = 1$, $R_2 = 3$, $R_3 = 5$. After the first iteration, $r_{1,1} = 1$, $r_{1,2} = 0$ and  $r_{1,3} = 0$. After the second iteration, $r_{2,1} = 0$, $r_{2,1} = 3$ and  $r_{2,3} = 0$. After the final iteration, $r_{3,1} = 0$, $r_{3,2} = 4$ and  $r_{3,3} = 1$.}
  \medskip
  \small
  \label{fig:algovis}
\end{figure}

\begin{figure}[htbp]
  \centering
  \includegraphics[width=0.75\linewidth]{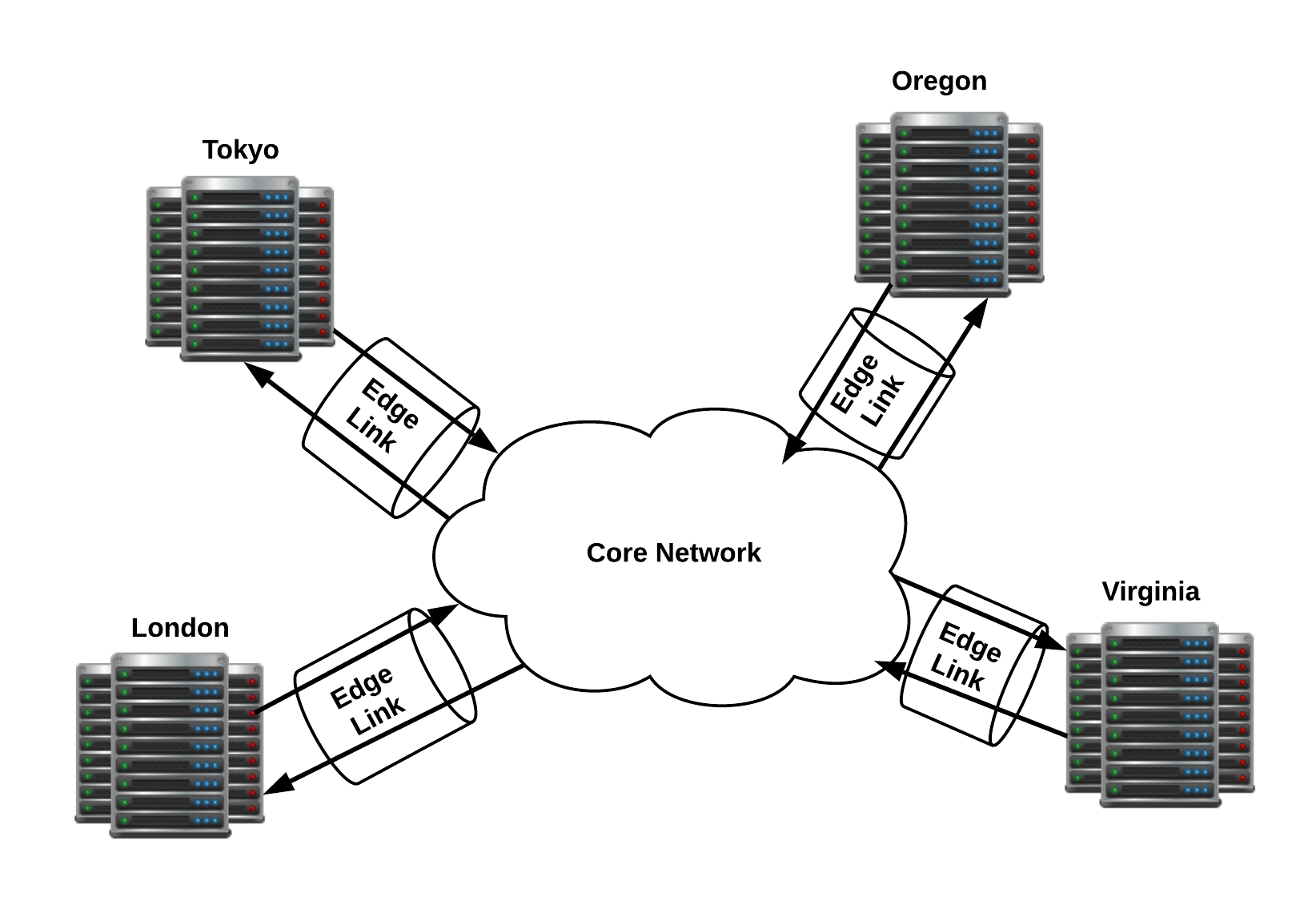}
  \caption{\textbf{Network Model}: An illustration of the network model used by this work}
  \medskip
  \small
  \label{fig:networkmodel}
\end{figure}

\begin{figure}[htbp]
  \centering
  \includegraphics[width=0.75\linewidth]{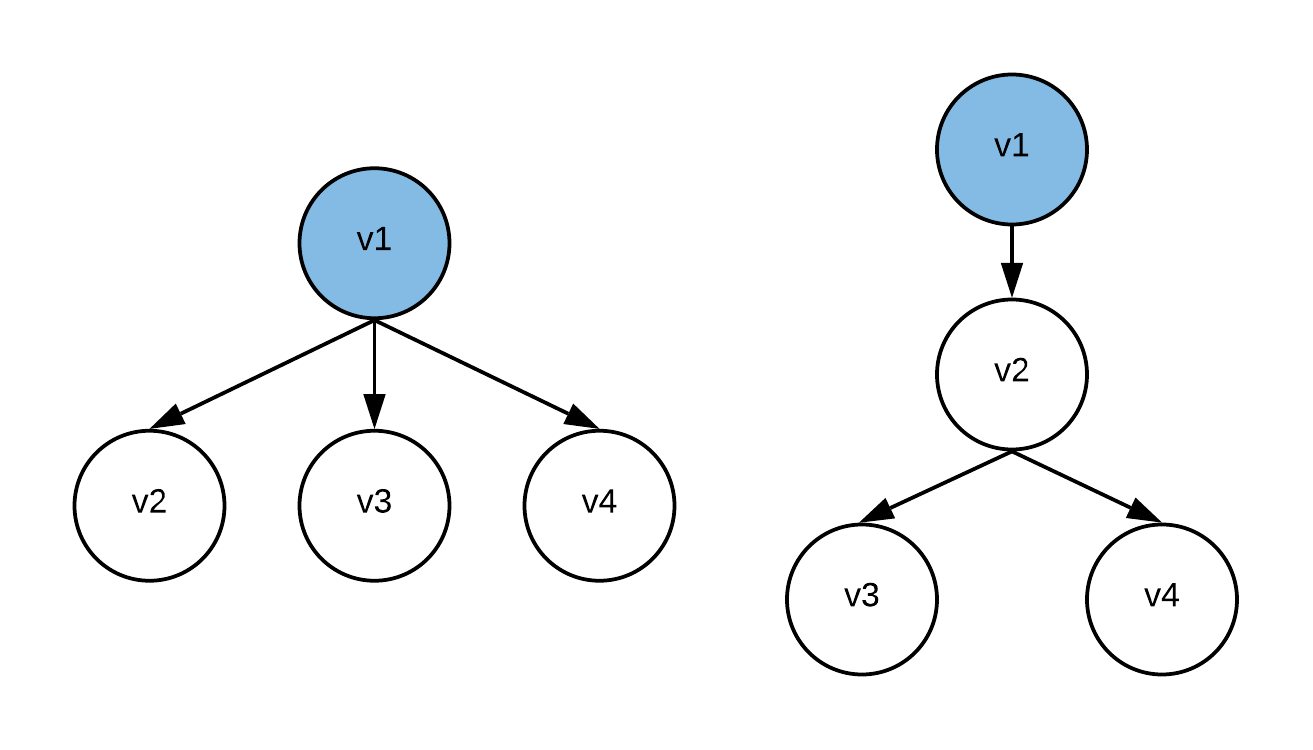}
  \caption{\textbf{Overlay Tree Examples}: This figure illustrates overlay trees constructed using the method described in section \ref{section:construct-substreams}. In a network with four sites, the tree on the left is the overlay tree of $s_{1,1}$ while the tree on the right is the overlay tree of $s_{1,2}$.}
  \medskip
  \small

  \label{fig:overlayexample}
\end{figure}

\newpage
\section{Proofs}
\subsection{Correctness Criteria}
The output of algorithm \ref{algo:1} is a set of sub-stream rates ${r_{1,1}, \ldots, r_{n,n}}$. A correct output satisfies the following three criteria for all $1 \leq i \leq n$:
\begin{enumerate}
    \item \textbf{Valid Partition Constraint: }The aggregated rate of all sub-streams of a client data stream equal to that client data stream's rate, which is equivalent to $\sum_{j=1}^{n} r_{i,j} = R_i$ for all $i$ such that $1 \leq i \leq n$.  
    \item \textbf{Uplink Capacity Constraint: } The aggregated rate of all sub-streams sent by $v_i$ is less than or equal to $C_u(v_i)$.  
    \item \textbf{Downlink Capacity Constraint: } The aggregated rate of all sub-streams received by $v_i$ is less than or equal to $C_d(v_i)$. 
\end{enumerate}

\subsection{Correctness of Algorithm \ref{algo:1}}
Let $U_i[\alpha]$ represent the value of $U_i$ at the start of iteration $\alpha$ of the outer loop
\begin{proposition}
\label{prop:1}
$\text{For all } \alpha \text{ such that } 1 \leq \alpha \leq n$, if $\text{ } \sum_{i = 1}^{n} U_i[\alpha] \geq (n-2)R_{\alpha}$, then $\sum_{i=1}^{n} r_{\alpha,i} = R_\alpha$ at the end of iteration $\alpha$ of outer loop.
\end{proposition}
\begin{proof}
This proposition can be proved by contradiction. Assume $\sum_{i = 1}^{n} U_i[\alpha] \geq (n-2)R_{\alpha}$ and $\sum_{i=1}^{n} r_{\alpha,i} \neq R_\alpha$ at the end of iteration $\alpha$ of the outer loop.  If that is the case, line \ref{line:10} is never executed,  otherwise,  $R^{\prime}_\alpha$ becomes 0 at line \ref{line:13} and the inner loop terminates with  $\sum_{i=1}^{n} r_{\alpha,i} = R_\alpha$. At the start of the inner loop's last iteration, $R^{\prime}_\alpha = R_\alpha - \sum_{i=1}^{n-1}r_{\alpha,i}$. Since line \ref{line:8} is executed at every iteration of the inner loop, we have $R^{\prime}_\alpha = R_\alpha - \frac{1}{n-2}\sum_{i=1}^{n-1}U_i[\alpha]$ and $(n-2)R^{\prime}_\alpha  > U_n[\alpha]$ at the start of the inner loop's last iteration.  Because $(n-2)R^{\prime}_\alpha  = (n-2)R_{\alpha} - \sum_{i=1}^{n-1}U_i[\alpha]$, $(n-2)R^{\prime}_\alpha  > U_n[\alpha]$ implies $(n-2)R_{\alpha} > \sum_{i = 1}^{n} U_i[\alpha]$,  which contradicts with the assumption.
\end{proof}

\begin{proposition}
\label{prop:2}
$\sum_{i = 1}^{n} U_i[\alpha] \geq (n-2)\sum_{i=\alpha}^{n}R_i,  \text{for all } \alpha \text{ such that } 1 \leq \alpha \leq n$
\end{proposition}
\begin{proof}
This proposition can be proved by induction on $\alpha$. \\ Base case $\alpha = 1$: $\sum_{i = 1}^{n} U_i[1] = \sum_{i = 1}^{n}(C_u(v_i) - R_i)$. According to condition (3) of sustainable rates, $\sum_{i = 1}^{n}(C_u(v_i) - R_i) \geq (n-2)\sum_{i = 1}^{n}R_i$.\\
Induction step: For an arbitrary number $\alpha$ such that $1 < \alpha \leq n$, assume the proposition \ref{prop:2} holds for $\alpha - 1$.  We have:
\begin{equation}
\label{eq:1}
\sum_{i = 1}^{n} U_i[\alpha - 1] \geq (n-2)\sum_{i=\alpha - 1}^{n}R_i
\end{equation}
As a result of proposition \ref{prop:1}, $\sum_{i=1}^{n} r_{\alpha - 1,i} = R_{\alpha - 1}$ at the end of the outer loop's iteration $\alpha - 1$. Since line \ref{line:11} is executed at every iteration of the inner loop, we have $U_i[\alpha] = U_i[\alpha-1] - (n-2)r_{\alpha-1,i}$ for all i such that $1 \leq i \leq n$. By summing up all $i$,  $\sum_{i = 1}^{n} U_i[\alpha] = \sum_{i = 1}^{n} U_i[\alpha - 1] - (n-2)\sum_{i=1}^{n} r_{\alpha - 1,i}$ and  $\sum_{i=1}^{n} r_{\alpha - 1,i} = R_{\alpha - 1}$. Therefore, by subtracting $(n -2)R_{\alpha - 1}$ from both sides of inequiality \ref{eq:1}, we have $\sum_{i = 1}^{n} U_i[\alpha] \geq (n-2)\sum_{i=\alpha}^{n}R_i$.
\end{proof}

\begin{proposition}
\label{prop:3}
The output of algorithm \ref{algo:1} satisfies Valid Partition Constraint
\end{proposition}
\begin{proof}
This proposition holds as the direct outcome of proposition \ref{prop:1} and  proposition \ref{prop:2}. 
\end{proof}

\begin{lemma}
\label{lemma:2}
$U_i[\alpha] \geq 0$ for all $i, \alpha$ such that $1 \leq i, \alpha \leq n$
\end{lemma}

\begin{proof}
 We prove this lemma by induction on $\alpha$. Base case $\alpha = 1$: By line \ref{line:2}, we have $U_i[1] = C_u(v_i) - R_i$ for all i such that $1 \leq i \leq n$, according to condition (1) of sustainable rates,  $U_i[1] \geq 0$ holds for all i such that $1 \leq i \leq n$.
 Induction step: For an arbitrary number $\alpha$ such that $1 < \alpha \leq n$, assume $U_i[\alpha - 1] \geq 0$. Line \ref{line:7} and line \ref{line:8} will guarantee $r_{\alpha,i} \leq \frac{U_i[\alpha - 1]}{n-2}$ and $U_i[\alpha] =  U_i[\alpha - 1] - (n-2)r_{\alpha,i}$ according to line \ref{line:11}. As a result, $U_i[\alpha] \geq 0$.
\end{proof}

\begin{proposition}
The output of algorithm \ref{algo:1} satisfies the Uplink Capacity Constraint
\end{proposition}
\begin{proof}
From lemma \ref{lemma:2}, we have $U_i \geq 0$ throughout the execution of algorithm \ref{algo:1} for all $i$ such that $1 \leq i \leq n$.
According to the overlay defined in the previous section \ref{section:construct-substreams}, sending $s_{i,j}$ consumes $r_{i,j}$ of $U_i^{\prime}$ and $U_i^{\prime}$ is consumed only by sending $v_i$'s sub-streams. By proposition \ref{prop:3} , $\sum_{j=1}^{n} r_{i,j} = R_i$ for all $i$ such that $1 \leq i \leq n$. Since $U_i^{\prime}$ equals to $R_i$, sending all of $v_i$'s sub-streams will consume exactly the amount of its reserved uplink bandwidth. Because $C_u(v_i) = U_i + U^{\prime}_i$, no uplink bandwidth constraint is violated.
\end{proof}

\begin{proposition}
The output of algorithm \ref{algo:1} satisfies the Downlink Capacity Constraint
\end{proposition}
\begin{proof}
 Since every sub-stream is broadcast by an overlay tree, each site receives all other site's data exactly once. According to the condition (2) of sustainable rates, there is also no violation of downlink bandwidth constraint. 
\end{proof}

%
\end{document}